\documentclass[english]{article}
\usepackage{lmodern}

\usepackage[T1]{fontenc}
\usepackage[latin9]{inputenc}
\usepackage{color}
\definecolor{note_fontcolor}{rgb}{0.80078125, 0.80078125, 0.80078125}
\usepackage{babel}
\usepackage{amsthm}
\usepackage{amsmath}
\usepackage{amssymb}
\usepackage[unicode=true,pdfusetitle,
 bookmarks=true,bookmarksnumbered=false,bookmarksopen=false,
 breaklinks=true,pdfborder={0 0 0},backref=false,colorlinks=true]
 {hyperref}

\makeatletter

\newcommand{\noun}[1]{\textsc{#1}}
\let\SF@@footnote\footnote
\def\footnote{\ifx\protect\@typeset@protect
    \expandafter\SF@@footnote
  \else
    \expandafter\SF@gobble@opt
  \fi
}
\expandafter\def\csname SF@gobble@opt \endcsname{\@ifnextchar[%]
  \SF@gobble@twobracket
  \@gobble
}
\edef\SF@gobble@opt{\noexpand\protect
  \expandafter\noexpand\csname SF@gobble@opt \endcsname}
\def\SF@gobble@twobracket[#1]#2{}
\newenvironment{lyxgreyedout}
  {\textcolor{note_fontcolor}\bgroup\ignorespaces}
  {\ignorespacesafterend\egroup}

\theoremstyle{plain}
\newtheorem{thm}{\protect\theoremname}[section]
  \theoremstyle{definition}
  \newtheorem{defn}[thm]{\protect\definitionname}
  \theoremstyle{remark}
  \newtheorem{rem}[thm]{\protect\remarkname}
\newcommand{\code}[1]{\texttt{#1}}

\makeatother

  \providecommand{\definitionname}{Definition}
  \providecommand{\remarkname}{Remark}
\providecommand{\theoremname}{Theorem}

\begin{document}
\global\long\def\NOOP{\mathbf{NOOP}}
\global\long\def\dom#1{\mathcal{#1}}
\global\long\def\fb#1{\dom{#1}}
\global\long\def\strfunarr{\multimap\!\rightarrow}
\global\long\def\COOP{\mathbf{COOP}}

\title{\label{cha:Domain-Theory}Domain Theory for Modeling OOP: A Summary\\
(Domain Theory for The Construction of $\NOOP$, and The Construction
of $\COOP$ as a Step Towards Constructing $\NOOP$)}

\maketitle
\noindent \begin{flushright}
\emph{}
\par\end{flushright}
\begin{abstract}
Domain theory is `a mathematical theory that serves as a foundation
for the semantics of programming languages'~\cite{Abramsky94}.
Domains form the basis of a theory of partial information, which extends
the familiar notion of partial function to encompass a whole spectrum
of \textquotedblleft degrees of definedness\textquotedblright , so
as to model \emph{incremental higher-order computation} (\emph{i.e.},
computing with infinite data values, such as functions defined over
an infinite domain like the domain of integers, infinite trees, and
such as objects of object-oriented programming%
\footnote{Objects of OOP are typically infinite data values because they are
usually recursively-defined via their definitions using the special
self-referential variables ``\code{this}'' or ``\code{self}''.%
}). General considerations from recursion theory dictate that partial
functions are unavoidable in any discussion of computability. Domain
theory provides an appropriately abstract setting in which the notion
of a partial function can be lifted and used to give meaning to higher
types, recursive types, etc.

$\NOOP$ is a domain-theoretic model of nominally-typed OOP~\cite{NOOP,NOOPbook,NOOPsumm,InhSubtyNWPT13,AbdelGawad14}.
$\NOOP$ was used to prove the identification of inheritance and subtyping
in mainstream nominally-typed OO programming languages and the validity
of this identification. In this report we first present the definitions
of basic domain theoretic notions and domain constructors used in
the construction of $\NOOP$, then we present the construction of
a simple structural model of OOP called $\COOP$ as a step towards
the construction of $\NOOP$. Like the construction of $\NOOP,$ the
construction of $\COOP$ uses earlier presented domain constructors.
\end{abstract}

\section{\label{sec:Domain-Theoretic-Notions}Basic Domain Theory Notions}

Domain theory is a branch of mathematics that builds on set theory,
order theory (\emph{i.e., }the theory of partially-ordered sets, \emph{a.k.a.},
posets), and topology (\emph{i.e.}, the theory of topological spaces).
It is relatively easy to digest the basic definitions of domain theory
once the computational motivations behind these definitions are understood.
Standard references on set theory include~\cite{Breuer06,Enderton77,Halmos60}.
Standard references on order theory include~\cite{DaveyPriestley90,OrderedSets05}.%
\footnote{Chapter 5 in~\cite{Turbak2008} presents an excellent introduction
to fixed points--a central topic in order theory and domain theory---that
is particularly suited for mathematically-inclined programmers.%
} Gierz, et al,~\cite{Gierz2003}, present a detailed encyclopaedic
account of domain theory, connecting domain theory to order theory
and to topology.

Otherwise, literature on domain theory is somewhat fractured. Terminology
in domain theory is somewhat less standard than that of set theory
and order theory. Accordingly, there is no standard formulation of
domain theory. Literature on domain theory includes~\cite{DTAL,Stoy77,ScottCompMathTheoryLects,Scott82,PlotkinDomains83,Allison86,DTMonograph,GunterHandbook90,KahnConDoms93,Abramsky94}.
Stoy's book~\cite{Stoy77} is a particularly detailed account of
the motivations behind domain theoretic definitions (Stoy, following
Scott's original formulation~\cite{DTAL}, uses complete lattices,
rather than cpos, for domains.)

In this and the next section we present the definitions of basic
domain theory notions used in constructing $\NOOP$ and $\COOP$.
In Section~\ref{sub:Domain-Constructor-Definitions} we present the
definitions of the domain constructors\emph{ }used in the constructions.\global\long\def\X{\fb X}
\global\long\def\S{\fb S}

\begin{defn}[Partial Order]
A \emph{partial order} (also called a partially-ordered set, or,
for short, a poset) is a pair $(\X,\sqsubseteq)$ consisting of a
set $\X$ (called the \emph{universe} of the ordering), and a binary
relation $\sqsubseteq$ on the set $\X$, such that\end{defn}
\begin{itemize}
\item $\forall x\in\X,x\sqsubseteq x$ ($\sqsubseteq$ is reflexive)
\item $\forall x,y\in\X,x\sqsubseteq y\wedge y\sqsubseteq x\implies x=y$
($\sqsubseteq$ is antisymmetric)
\item $\forall x,y,z\in\X,x\sqsubseteq y\wedge y\sqsubseteq z\implies x\sqsubseteq z$
($\sqsubseteq$ is transitive)
\end{itemize}
where $\implies$ is implication. The relation $\sqsubseteq$ is usually
called the `less than or equals' relation when discussing general
posets, and is called the `approximates' relation in domain theory.
Intuitively, $x\sqsubseteq y$ means $x$ is `no more informative
than' (\emph{i.e.}, approximates information contained in) $y$. A
poset $(\X,\sqsubseteq)$ is usually referred to using the symbol
for its universe, $\X$. We do so below. When we need to specifically
refer to the universe, \emph{i.e.}, the set underlying a poset $\X$,
we instead use the bar notation $|\X|$ to denote this universe.
\begin{rem}
In domain theory, the approximation ordering is defined on mathematical
values used to denote \emph{computational} data values. The approximation
ordering has intuitive connections to information theory. A computational
value whose denotation approximates the denotation of another computational
value is considered no more informative than the second data value.
The approximation ordering is a qualitative expression of the relative
informational content of computational values (which are denoted by
elements of the universe of the ordering). Computational values whose
denotations are higher in the approximation ordering are more informative
than ones whose denotations are lower in the ordering.
\end{rem}
\smallskip{}

\begin{rem}
The least computational value is divergence (as in an `infinite loop').
It gives no information, and thus is the least informative computational
value. Given that divergence gives no information, the abstract mathematical
value denoting divergence is called `bottom', is at the bottom of
the approximation ordering (hence the name), and is usually denoted
by the symbol $\bot$.\end{rem}
\begin{defn}[Induced Partial Order]
Every subset $\S$ of the universe of a poset $\X$ has an associated
partial order called the \emph{induced} \emph{partial order} of $\S$.
Members of the ordering relation of the induced order are those of
the ordering of $\X$ restricted to elements $\S$.\end{defn}
\begin{rem}
The induced partial order of a subset of a poset $\X$ is sometimes
called a \emph{subposet} of $\X$. In a usually-harmless and standard
abuse of terminology and notation, we refer to induced partial orders
as \emph{subsets} instead, and we use $\S$ to denote both the subset
and its induced partial order.
\end{rem}
\smallskip{}

\begin{defn}[Upper bound]
Given a subset $\S$ of a poset $\X$, an \emph{upper bound} of $\S$,
in $\X$, is an element $x\in\X$ such that $\forall s\in\S,s\sqsubseteq x$.
\end{defn}
\smallskip{}

\begin{defn}[Bounded]
A subset $\S$ of a poset $\X$ is \emph{bounded} \emph{in $\X$}
iff $\S$ has an upper bound in $\X$.
\end{defn}
\smallskip{}

\begin{defn}[Least Upper Bound]
An upper bound of a subset $\S$ in a poset $\X$ is a \emph{least
upper bound} (also called a \emph{lub}, or LUB) of $\S$ iff this
upper bound approximates all upper bounds of $\S$ in $\X$. If it
exists, the lub of $\S$ is denoted $\bigsqcup\S$.%
\footnote{A lub of a subset $\S$ may not exist, either because $\S$ has no
upper bounds or because $\S$ has more than one upper bound but there
is no least element (\emph{i.e.}, a minimum) among them.%
}
\end{defn}
\smallskip{}

\begin{defn}[Downward-Closed]
A subset $\S$ of a poset $\X$ is a \emph{downward-closed }set iff
all elements $x$ of $\X$ that approximate some element in $\S$
belong to $\S$. Thus, $\S$ is downward-closed iff $\forall x\in\X.((\exists s\in\S,x\sqsubseteq s)\implies x\in\S)$.
\end{defn}
\smallskip{}

\begin{defn}[Chain]
A countable subset $\S$ of a poset $\X$ with elements $s_{i}$
is a \emph{chain} if $\forall i,j\in\mathbb{N}.i\leq j\rightarrow s_{i}\sqsubseteq s_{j}$.\end{defn}
\begin{rem}
Every finite chain includes its lub (the maximum element of the chain).
Infinite chains (like set $\mathbb{N}$ under the standard ordering)
do \emph{not} necessarily have maximal elements.\end{rem}
\begin{defn}[Anti-chain]
A countable subset $\S$ of a poset $\X$ with elements $s_{i}$
is an \emph{anti-chain} if $\forall i,j\in\mathbb{N}.i\neq j\rightarrow s_{i}\not\sqsubseteq s_{j}$.\end{defn}
\begin{rem}
A \emph{flat} poset $\fb R$ is an anti-chain $\S$ with elements
$s_{i}$ and an additional bottom element $\bot_{\fb R}$, such that
$\bot_{\fb R}\sqsubseteq s_{i}$ and $\bot_{\fb R}\neq s_{i}$ for
all $i$. A flat poset, thus, is said to be the \emph{lifting }of
the underlying anti-chain.\end{rem}
\begin{defn}[Directed]
A subset $\S$ of a poset $\X$ is \emph{directed} iff every finite
subset of $\S$ is bounded in $\S$.\end{defn}
\begin{rem}
\label{Rem:Chain-Dir}Every chain is a directed set, but not necessarily
vice versa.\end{rem}
\begin{defn}[Consistent]
\label{Defn:Consistent}A subset $\S$ of a poset $\X$ is \emph{consistent
in $\X$} iff every finite subset of $\S$ is bounded in $\X$.\end{defn}
\begin{rem}
In general posets, every bounded set is consistent, but not necessarily
vice versa. Consistency requires the boundedness of \emph{finite}
subsets only. Thus, boundedness (where \emph{all} subsets are bounded)
is a stronger condition than consistency.
\end{rem}
\smallskip{}

\begin{rem}
Because $\S$ is a subset of $\X$, boundedness in $\S$ implies boundedness
in $\X$, and thus every directed set $\S$ is a consistent set, but
not necessarily vice versa. Directedness is thus also a stronger condition
than consistency.
\end{rem}
\smallskip{}

\begin{defn}[Ideal]
\label{Defn:Ideal}A subset $\S$ of a poset $\X$ is an \emph{ideal
}iff it is downward-closed and directed.
\end{defn}
\smallskip{}

\begin{defn}[Lower set]
A subset $\S_{x}$ of a poset $\X$ is a \emph{lower set }of an element
$x\in|\X|$ iff it contains all elements of $|\X|$ that are less
then or equal to $x$ (and nothing else). Thus, for $x\in|\X|$, $\S_{x}$
is the lower set of $x$ iff $\S_{x}=\{s\in|\X||s\sqsubseteq x\}$.
\end{defn}
\smallskip{}

\begin{defn}[Principal Ideal]
A subset $\S_{x}$ of a poset $\X$ is a \emph{principal ideal} (determined
by $x)$ iff it is the lower set of $x$.\end{defn}
\begin{thm}
(Principal Ideals are Ideals) A subset $\S$ of a poset $\X$ is an
ideal if it is a principal ideal.\end{thm}
\begin{proof}
Note that, by definition and using the transitivity of $\sqsubseteq$,
a lower set of an element $x\in\X$ is downward-closed. The lower
set of $x$ is also directed because it contains $x$ and $x$ is
a bound for all (finite) subsets of the lower set.\end{proof}
\begin{defn}[Weak Ideal]
A non-empty subset $\S$ of a poset $\X$ is a \emph{weak ideal}
iff\emph{ }it is downward-closed and is closed under lubs of its chains.\end{defn}
\begin{rem}
Every flat poset is a weak ideal. Chains in flat posets have two
elements, the lower of which is always $\bot$.
\end{rem}
\smallskip{}

\begin{defn}[Finitary Basis]
A poset $\X$ is a \emph{finitary basis} iff its universe, $\left|\X\right|$,
is countable and every finite bounded subset $\S$ of $\X$ has a
lub in $\X$.\end{defn}
\begin{rem}
From the definition of finitary basis, the fact that a finite subset
$\S$ of a finitary basis $\X$ is bounded is equivalent to $\S$
having a lub. Generally, this statement is true only in one direction
for an arbitrary poset (\emph{i.e.}, the trivial $\Leftarrow$ direction,
which asserts the boundedness of a set if it has a lub.) In a finitary
basis, the opposite direction is true as well for all finite subsets
of the finitary basis.\end{rem}
\begin{defn}[Complete Partial Order]
\label{Defn:dCPO}A poset $\X$ is a \emph{complete partial order
(cpo, or, sometimes, dcpo)} iff every directed subset $\S$ of $\X$
has a lub in $\X$, \emph{i.e.}, a cpo is closed over lubs of its
directed subsets.\end{defn}
\begin{thm}[Ideals over a FB form a cpo]
\label{thm:ideals-dcpo}Given a finitary basis $\X$, the set $\dom I_{\X}$
of ideals of $\X$ is a cpo \emph{under the subset ordering $\subseteq$}.\end{thm}
\begin{proof}
Under the subset ordering, a directed set $\dom J$ of ideals of $\X$
is one in which each finite subset $\dom J_{f}$ of $\dom J$ has
an element in $\dom J$ (\emph{i.e.}, an ideal) that includes all
elements in the elements of $\dom J_{f}$. Every such directed set
$\dom J$ has a lub in $\X$ under the subset ordering, namely the
union of elements of $\dom J$, $\bigcup\dom J$. This union is always
an ideal, and thus a member of $\dom I_{\X}$.%
\footnote{A Coq~\cite{Bertot2004} development (\emph{i.e.}, a Coq proof script)
with a proof of this theorem is available upon request.%
}%
\begin{lyxgreyedout}
\end{lyxgreyedout}
\end{proof}
\begin{defn}[Constructed Domain]
Given a finitary basis $\X$, the set $\dom I_{\X}$, of ideals of
$\X$,  forms a poset, $(\dom I_{\X},\subseteq)$ is called the
\emph{domain determined by }$\X$ or, sometimes, the \emph{ideal completion}
of $\X$. $\dom{I_{X}}$ is, thus, called a \emph{constructed domain}
(\emph{i.e.}, one that is defined by the finitary basis $\X$).\end{defn}
\begin{rem}
By Theorem~\ref{thm:ideals-dcpo}, the ideal completion of (\emph{i.e.},
the domain determined by) every finitary basis is a cpo.
\end{rem}
\smallskip{}

\begin{defn}[Finite Element of a CPO]
An element $d$ of a cpo $\dom D$ is a \emph{finite} \emph{element}
(or, equivalently, \emph{isolated }or \emph{compact}) iff $d$ belongs
to each directed subset $\S$ that $d$ is a lub of. The set of finite
elements of a cpo $\dom D$ is denoted by $\dom D^{0}$.%
\footnote{This definition of finite elements is weaker than the usual definition
for cpos. In the context of domains, which are finitary-based, the
two definitions are equivalent.%
}
\end{defn}
\smallskip{}

\begin{defn}[Isomorphic Partial Orders]
Two posets are \emph{isomorphic} iff there is an order-preserving
one-to-one onto function between them.
\end{defn}
\smallskip{}

\begin{defn}[Domain]
A cpo $\dom D$ is a \emph{domain} iff its finite elements $\dom D^{0}$
form a finitary basis and $\dom D$ is isomorphic to the domain determined
by the finitary basis $\dom D^{0}$.
\end{defn}
\smallskip{}

\begin{defn}[Subdomain]
\label{Subdomain}As a counterpart to the notion of subset in set
theory, and subposet in order theory, a domain $\dom D$ is a \emph{subdomain}
of a domain $\dom E$ iff (1) their universes are in the subset relation,$\left|\dom D\right|\subseteq\left|\dom E\right|$,
(2) they have the same bottom element,$\bot_{\dom D}=\bot_{\dom E}$,
(3) restricted to elements of their respective universes, they have
the same approximation ordering,$\forall d_{1},d_{2}\in\dom D,d_{1}\sqsubseteq_{\dom D}d_{2}\Leftrightarrow d_{1}\sqsubseteq_{\dom E}d_{2}$
(\emph{i.e.}, approximation ordering for $\dom D$ is the approximation
ordering of $\dom E$ restricted to elements of $\dom D$), and (4)
restricted to elements of their respective universes, they have the
same lub relation, $\forall d_{1},d_{2},d_{3}\in\dom D,\left(d_{1}\sqcup_{\dom D}d_{2}=d_{3}\right)\Leftrightarrow\left(d_{1}\sqcup_{\dom E}d_{2}=d_{3}\right)$
(\emph{i.e.}, the lub relation for $\dom D$ is the lub relation of
$\dom E$ restricted to elements of $\dom D$).\end{defn}
\begin{rem}
For a subdomain $\dom D$ of domain $\dom E$, the domain determined
by $\dom D^{0}$ is isomorphic to the domain determined by $\dom E^{0}\cap\dom D$
(which must be a finitary basis.)
\end{rem}
\smallskip{}

\begin{rem}
In Definition~\ref{Subdomain}, we use Scott's definition of subdomains
because we define $\NOOP$ and $\COOP$ domains as subdomains of Scott's
universal domain $\dom U$. Scott~\cite{ScottCompMathTheoryLects,DTMonograph}
shows that every domain is isomorphic to a subdomain of $\dom U$.�
Under the subdomain ordering, all the subdomains of $\dom U$ form
a domain (itself also a subdomain of $\dom U$, by the universality
of $\dom U$.)� All domains given in a domain equation and all recursively
defined domains in the equation are elements of this space of domains
(again, a domain that consists of all of the subdomains of $\dom U$
as its elements).�� Thus, solutions of recursive domain equations
(as elements of the domain of subdomains of $\dom U$) are defined
in the same way (\emph{e.g.}, as \emph{least fixed-points}, or \emph{lfp}s
of generating functions) as solutions of recursive definitions specifying
elements in any other computational domain (a subdomain of $\dom U$).
\end{rem}

\section{\label{sub:Definitions-Functional-Domains}Notions for Functional
Domains}

To model computable functions, domain theory provides functional domains,
whose elements are particular mathematical functions mapping elements
from one computational domain to another. To define functional domains,
we will introduce the domain theoretic notions of `approximable mappings'
(AMs), `finite-step mapping', and `continuous functions'.
\begin{defn}[Approximable Mapping]
\label{AM-conditions}Given two finitary basis $A$ and $B$, with
ordering relations $\sqsubseteq_{A}$ and $\sqsubseteq_{B}$, respectively,
a relation $\mathsf{f}_{am}\subseteq\left|A\right|\times\left|B\right|$
is an \emph{approximable mapping} (AM) iff\end{defn}
\begin{enumerate}
\item \label{enu:AM-Condition-bot}Condition~\ref{enu:AM-Condition-bot}:
$(\bot_{A},\bot_{B})\in\mathsf{f}_{am}$ (\emph{pointedness})
\item \label{enu:AM-Condition-dc}Condition~\ref{enu:AM-Condition-dc}:
$\forall a\in A.\forall b_{1},b_{2}\in B.\bigl((a,b_{2})\in\mathsf{f}_{am}\,\wedge\, b_{1}\sqsubseteq_{B}b_{2}\,\rightarrow\,(a,b_{1})\in\mathsf{f}_{am}\bigr)$
(\emph{downward-closure})
\item \label{enu:AM-Condition-dir}Condition~\ref{enu:AM-Condition-dir}:
$\forall a\in A.\forall b_{1},b_{2}\in B.\bigl((a,b_{1})\in\mathsf{f}_{am}\,\wedge\,(a,b_{2})\in\mathsf{f}_{am}\,\rightarrow\,(a,\, b_{1}\sqcup_{B}b_{2})\in\mathsf{f}_{am}\bigr)$
(\emph{directedness})
\item \label{enu:AM-Condition-mon}Condition~\ref{enu:AM-Condition-mon}:
$\forall a_{1},a_{2}\in A.\forall b\in B.\bigl((a_{1},b)\in\mathsf{f}_{am}\,\wedge\, a_{1}\sqsubseteq_{A}a_{2}\,\rightarrow\,(a_{2},\, b)\in\mathsf{f}_{am}\bigr)$
(\emph{monotonicity})\end{enumerate}
\begin{defn}[Set Image under a Relation]
\label{Defn:Set-Image}Given sets $A$, $B$ and a relation $r\subseteq A\times B$,
the \emph{set image} of a subset $\S$ of $A$ under $r$, denoted
by $r(\S)$, is the set of all $b\in B$ related in $r$ to some element
in $\S$. Hence, relation $r$ is viewed as a function over subsets
of $A$. For $\S\subseteq A$, we have $r(\S)=\{b\in B|\exists a\in\S.(a,b)\in r\}$.
The set image of a relation $r$ also allows viewing $r$ as a function
$r:A\rightarrow\wp(B)$, where $r(a)=r(\{a\})$ for $a\in A$. In
other words, for $a\in A$, function $r$ returns the set of all $b\in B$
related to $a$ in $r$ (viewed as a relation).\end{defn}
\begin{thm}[AMs map ideals to ideals]
\label{AMs-map-ideals-to-ideals}Given finitary basis $A$ and $B$,
if $\mathsf{f}_{am}$ is an approximable mapping from $A$ to $B$,
and if $I$ is an ideal in $A$, then \textup{$\mathsf{f}_{am}(I)$,
the set image of $I$ under $\mathsf{f}_{am}$, is an ideal in $B$.}\end{thm}
\begin{proof}
From the definition of an ideal, and using AM Condition~\ref{enu:AM-Condition-dc}
(which guarantees the set image is downward-closed), and AM Condition~\ref{enu:AM-Condition-dir}
(which guarantees the set image is directed).\end{proof}
\begin{thm}[AMs are monotonic]
Given finitary basis $A$ and $B$, if $\mathsf{f}_{am}$ is an approximable
mapping from $A$ to $B$, and if $I_{1}$ and $I_{2}$ are ideals
in $A$ such that $I_{1}\subseteq I_{2}$, then \textup{$\mathsf{f}_{am}(I_{1})\subseteq\mathsf{f}_{am}(I_{2})$
in $B$.}\end{thm}
\begin{proof}
By AM Condition~\ref{enu:AM-Condition-mon}.\end{proof}
\begin{defn}[Finite-Step Mapping]
\label{Defn:FSMs}Given finitary basis $A$ and $B$, an approximable
mapping $\mathsf{f}_{am}$ is a \emph{finite-step mapping} iff it
is the smallest approximable mapping containing some finite subset
of $\left|A\right|\times\left|B\right|$.
\end{defn}
\smallskip{}

\begin{defn}[Continuous Function]
\label{Defn:Continuous-Function}Given domains $\dom A$ and $\dom B$,
a function $f:\dom A\rightarrow\dom B$ from domain $\dom A$ to $\dom B$
is a \emph{continuous function} iff the value of $f$ at the lub of
a directed set of $a$'s in $\dom A$ is the lub, in $\dom B$, of
the (directed) set of function values $f(a)$.\end{defn}
\begin{rem}
Continuity of a function requires the value of the function at a
(non-finite) limit point $l$ to equal the limit of values of the
function at the finite approximations to $l$. Continuous functions
are thus said to ``have no surprises at the limit''.
\end{rem}
\smallskip{}

\begin{rem}
\label{AM-cont}Because of the four AM conditions, if finitary basis
$A$ and $B$ determine domains $\dom A$ and $\dom B$, respectively
, then every approximable mapping in $\left|A\right|\times\left|B\right|$
determines a continuous function in $\dom A\rightarrow\dom B$, \emph{and
vice versa}. Check Cartwright and Parsons' `Domain Theory: An Introduction'
monograph~\cite{DTMonograph} and other domain theory literature
for proof and more details.
\end{rem}
\smallskip{}

\begin{rem}
To motivate the preceding definitions, it should be noted that continuous
functions capture the fact that computation is of a ``finitely-based''
nature. Only finite data values can have canonical representations
inside a computing device. From a domain-theoretic perspective, an
(infinite) function can be computable only if its value ``at infinity''
(\emph{i.e.}, at an infinite input data value) is the one we  expect
by only seeing (and extrapolating from) the values of the function
at all finite inputs that approximate the infinite input data value
(finite inputs are all that can be represented inside computers, and
thus they are all that can be computed with). See Stoy's book~\cite{Stoy77}
for more details on motivation and intuitions behind domain theoretic
definitions.\emph{}%
\footnote{Via Roger's work, Dana Scott managed to connect the notion of continuous
functions to the notion of computable functions in computability theory.
Again, see Stoy's book~\cite{Stoy77} for more details.%
}
\end{rem}
\smallskip{}

\begin{rem}
Approximable mappings offer the means to accurately characterize
and define continuous functions (which, as mentioned above, capture
the finitely-based nature of computation). Finite-step mappings, as
the ``finite/representable parts'' of AMs, offer the means by which
continuous functions can be constructed from more elementary parts
that can be represented in a computing device.
\end{rem}

\section{\label{sub:Domain-Constructor-Definitions}Domain Constructors}

In this section we present the domain constructors used to define
$\NOOP$ and $\COOP$.

\subsection{\label{sub:Coalesced-Sum}Coalesced Sum ($+$)}

The first domain constructor we present is the \emph{coalesced sum}
domain constructor, $+$. The expression $\fb A+\fb B$ denotes the
coalesced sum of two domains $\fb A$ and $\fb B$, with approximation
ordering relations $\sqsubseteq_{\fb A}$ and $\sqsubseteq_{\fb B}$,
respectively. A coalesced sum is a domain-theoretic counterpart of
the standard set-theoretic disjoint union operation.

If $\fb C=\fb A+\fb B$ then 
\[
\left|\fb C\right|=\{\bot_{\mathsf{C}}\}\cup\{(0,a)|a\in(\left|\fb A\right|\backslash\{\bot_{\fb A}\})\}\cup\{(1,b)|b\in(\left|\fb B\right|\backslash\{\bot_{\fb B}\})\}
\]
where $0$ and $1$ are used in $\fb C$ to tag non-bottom elements
from $\fb A$ and $\fb B$, respectively.

The ordering relation $\sqsubseteq_{\fb C}$, on elements of $\fb C$,
is defined, for all $c_{1},c_{2}\in\fb C$, by the predicate
\begin{eqnarray}
c_{1}\sqsubseteq_{\fb C}c_{2} & \Leftrightarrow & (c_{1}=\bot_{\fb C})\vee(c_{1}=(0,a_{1})\wedge c_{2}=(0,a_{2})\wedge a_{1}\sqsubseteq_{\fb A}a_{2})\nonumber \\
 &  & \phantom{(c_{1}=\bot_{\fb C)}}\vee(c_{1}=(1,b_{1})\wedge c_{2}=(1,b_{2})\wedge b_{1}\sqsubseteq_{\fb B}b_{2})\label{eq:summation-approx}
\end{eqnarray}

\subsection{\label{sub:Strict-Product}Strict Product ($\times$)%
\footnote{\label{fn:using-times}In agreement with the standard convention in
domain theory literature, the symbol $\times$ is overloaded in this
report. The symbol $\times$ is used to denote the strict product
of ordered sets (including domains), and is also used to denote the
standard set-theoretic cross product (which ignores any ordering on
its input sets). It should always be clear from context which meaning
is attributed to $\times$.%
}}

We use $\fb A\times\fb B$ to denote the strict product of two domains,
$\fb A$ and $\fb B$, with approximation ordering relations $\sqsubseteq_{\fb A}$
and $\sqsubseteq_{\fb B}$, respectively. A strict product is an order-theoretic
counterpart of the standard set-theoretic cross-product operation.

If $\fb C=\fb A\times\fb B$ then 
\begin{equation}
\left|\fb C\right|=\mathsf{(\left|\mathsf{A}\right|\backslash\{\bot_{\fb A}\})}\times\mathsf{(\left|\mathsf{B}\right|\backslash\{\bot_{\fb B}\})}\cup\{\bot_{\fb C}\}\label{eq:strict-product-set}
\end{equation}
Strictness of $\times$ means that in $\fb C$, $\bot_{\fb C}$ replaces
all\emph{ }pairs $(a,b)\in\fb A\times\fb B$ where $a=\bot_{\fb A}$
or $b=\bot_{\fb B}$. Similar to the definition of the coalesced sum
constructor, this strictness is achieved in the definition above by
excluding $\bot_{\fb A}$ and $\bot_{\fb B}$ from the input sets
of the set-theoretic cross product. Sometimes the strict product $\fb A\times\fb B$
is called their `smash product'.

The ordering relation $\sqsubseteq_{\fb C}$, on elements of $\fb C$,
is defined as follows. $\forall c_{1},c_{2}\in\fb C,\forall a_{1},a_{2}\in\fb A\backslash\{\bot_{\fb A}\},\forall b_{1},b_{2}\in\fb B\backslash\{\bot_{\fb B}\}$
where $c_{1}=(a_{1},b_{1})$ or $c_{1}=\bot_{\fb C}$, and $c_{2}=(a_{2},b_{2})$
or $c_{2}=\bot_{\fb C}$
\begin{equation}
c_{1}\sqsubseteq_{\fb C}c_{2}\Leftrightarrow(c_{1}=\bot_{\fb C}\vee(a_{1}\sqsubseteq_{\fb A}a_{2}\wedge b_{1}\sqsubseteq_{\fb B}b_{2})).\label{eq:strict-product-approx}
\end{equation}

\subsection{\label{sub:FSMs-(Finite-Step-Mappings)}Continuous Functions ($\rightarrow$)}

Functional domains and functional domain constructors are necessary
for accurately modeling OOP. Functional domains of $\NOOP$ (and $\COOP$)
are: (1) the auxiliary domain of methods whose members are strict
continuous functions modeling object methods, and (2) the auxiliary
domain of records, whose members are `record functions' modeling
record components of objects. (A record function, constructed using
a new domain constructor $\multimap$, called `rec', is a function
defined over a finite set of labels. See~\cite{NOOP,AbdelGawad14}
for the definition of the records domain constructor,$\multimap$,
and proofs of its properties. See~\cite{NOOPsumm} for a summary.)

The symbol $\rightarrow$ is used to denote the standard continuous
functions domain constructor. Making use of the definitions of domain
theoretic notions presented in Section~\ref{sub:Definitions-Functional-Domains},
particularly approximable mappings and finite-step mappings, we refer
the reader to Chapter~3 of Cartwright and Parsons' monograph on
Domain Theory~\cite{DTMonograph} (which is an update of Scott's
lecture notes~\cite{ScottCompMathTheoryLects}) for the details of
the definition of the continuous functions domain constructor $\rightarrow$.
Since there is a one-to-one correspondence between domains and their
finitary bases, and given that the latter are simpler and more intuitive
notions, Cartwright and Parsons' monograph describes how the domain
$\dom A\rightarrow\dom B$ of continuous functions from domain $\dom A$
to domain $\dom B$ is determined by constructing its finitary basis
from the finitary basis of domains $\dom A$ and $\dom B$ (See Remark~\ref{AM-cont}).

In this report, we use the symbol $\strfunarr$ to denote the \emph{strict}
continuous functions domain constructor, which simply constructs a
space like the space of continuous functions from domain $\dom A$
to domain $\dom B$ but where all so-called ``one-step functions''
of the form $\bot_{\dom A}\mapsto b$ (for $b\in\dom B\backslash\{\bot_{\dom B}\}$)
are eliminated (\emph{i.e.}, are mapped to the one-step function $\bot_{\dom A}\mapsto\bot_{\dom B}$,
which is the bottom element of the constructed function space.) Strict
continuous functions map $\bot_{\dom A}$ only to $\bot_{\dom B}$,
thereby modeling strict computable functions (\emph{i.e.}, functions
that have ``call-by-value'' semantics.)

A notable property of functional domain constructors is that the set
of continuous functions between two domains itself forms a domain.
This property (\emph{i.e.}, finding a mathematical space having this
property) has been much behind the development of domain theory.

\subsection{\label{sub:Sequences}Strict Finite Sequences $(\protect\fb D^{*})$}

For the purpose of constructing methods of $\NOOP$ and $\COOP$,
one more domain constructor is needed: the constructor of the domain
of strict finite sequences. This constructor is used to construct
the finite sequences of objects that are passed as arguments to methods
of objects in $\NOOP$ and $\COOP$. Sometimes the domain $\fb D^{*}$
of finite sequences of elements of domain $\fb D$ is called the Kleene
closure of domain $\fb D$.

The Kleene closure, $\fb D^{*}$, constructs a domain of finite sequences
of elements of its input domain, $\fb D$, including the empty sequence.
Our definition of $^{*}$ excludes constructing sequences of $\fb D$
where a member of the sequence is $\bot_{\fb D}$. Thus, $^{*}$ is
said to construct \emph{strict} finite sequences.

The Kleene closure is defined as a set of all $n$-tuples of elements
of $\fb D$ (where \emph{$n$} is a natural number). Thus
\[
\left|\fb D^{*}\right|=\{\bot_{\fb D^{*}}\}\cup\bigcup_{n\in\mathbb{N}}\{<d_{0},\cdots,d_{i},\cdots,d_{n-1}>|d_{i}\in(\left|\fb D\right|\backslash\{\bot_{\fb D}\})\}
\]
An element $u$ of $\fb D^{*}$ approximates an element $v$ of $\fb D^{*}$
iff $u=\bot_{\fb D^{*}}$ or the lengths of both $u$ and $v$ are
equal to a natural number $k$, and $u_{i}\sqsubseteq_{\fb D}v_{i}$
for all $0\leq i<k$.

\global\long\def\SOOP{\mathbf{SOOP}}

\section{\label{cha:COOP}$\protect\COOP$: A Simple Structural Model of OOP}

In this section we present the construction of $\COOP$ as a simple
structural domain-theoretic model of OOP. The reasons for constructing
a structural model of OOP, \emph{i.e.}, $\COOP$, as a step towards
construcing $\NOOP$ as a model of nominally-typed OOP are threefold.
First, (1) earlier research on structural OOP needs to be put on a
more rigorous footing. The literature on models of structural OOP
glosses over important technical details like the construction of
a domain of records, having methods of multiple arity, and objects
being purely OO (\emph{i.e.}, not allowing functions and non-object
values have first-class status in the constructed domain of ``objects''),
all of which we address in the construction of $\COOP$. Second, (2)
the construction of $\COOP$ is similar to but simpler than the construction
of $\NOOP$ (\emph{e.g.}, $\COOP$ ``objects'' do not include signatures,
and thus constructing $\COOP$ does not need an extra filtering step
to match signatures with record components of objects as is needed
for $\NOOP$ construction). Understanding how $\COOP$ is constructed
makes it easier to understand the construction of $\NOOP$. Third,
and most importantly, (3) the rigorous definition of $\COOP$ alongside
the definition of $\NOOP$ clarifies the distinction between structural
OOP and nominal OOP.

As mathematical models, $\COOP$ and $\NOOP$ are collections of semantics
domains. In denotational semantics, domains are used to model computational
constructs. Domains of $\COOP$ and $\NOOP$ correspond to the set
of all possible object values, field values, and method values of
structural and nominal OO programs, respectively. Similarly, specific
subdomains of $\COOP$ and $\NOOP$ domains correspond to specific
structural and nominal types definable in structurally-typed and nominally-typed
OO languages. $\COOP$ and $\NOOP$, thus, give an abstract mathematical
meaning to the most fundamental concepts of structurally-typed and
nominally-typed OOP.

Focusing on $\COOP,$ our presentation of $\COOP$ proceeds as follows.
The domain of objects of $\COOP$ is the solution of a reflexive domain
equation. In Section~\ref{sec:COOP-Domain-Equation} we first
present the $\COOP$ domain equation. In Section~\ref{sec:COOP-Construction}
we then show how $\COOP$ domains are constructed as the solution
of the $\COOP$ domain equation. The domains of $\COOP$ are constructed
using standard domain theoretic construction methods that make use
of standard domain constructors as well as the records domain constructor,
$\multimap$ (pronounced ``rec''), described in~\cite{NOOP,NOOPbook,NOOPsumm,AbdelGawad14}.

The view of objects in $\COOP$ is a very simple one. An object
in $\COOP$ is a record of functions that map sequences of objects
to objects. In other words, in $\COOP$ an object is `a finite collection
of methods', where a method is a labeled function mapping sequences
of objects to objects. (In $\COOP$, unlike $\NOOP$, we encode fields
as zero-ary methods.)

Given that it is a structural model of OOP, $\COOP$ closely resembles
$\SOOP$, the model of OOP Cardelli presented in~\cite{Cardelli84,Cardelli88}.
Given that objects of $\COOP$, like those of $\SOOP$, miss nominality
information, $\COOP$ is a structural model of OOP. Our presentation
of the construction of $\COOP$ in the following sections shows how
to rigorously construct a model like Cardelli's. $\COOP$, however,
differs from $\SOOP$ in five respects:
\begin{enumerate}
\item Unlike $\SOOP,$ but similar to many mainstream OO languages, the
$\COOP$ domain equation does not allow functions as first-class values
(thus, $\COOP$ does not support direct ``currying'' of functions).
Only objects are first-class values in $\COOP$.
\item Unlike $\SOOP$, $\COOP$ uses the records domain constructor, $\multimap$,
to construct records rather than using the standard continuous functions
domain constructor (which is used in $\SOOP$). The definition of
$\multimap$ is presented in~\cite{NOOP,NOOPbook,NOOPsumm,AbdelGawad14}.
\item Unlike $\SOOP$, methods in $\COOP$ objects are multi-ary functions
over objects.%
\footnote{Since $\SOOP$ defines a domain for a simple functional language with
objects based on \noun{ML}, it is natural to force all functions to
be unary (as in \noun{ML}). In this context, a multi-ary function
can be transparently curried.%
}
\item For simplicity, $\COOP$ objects have fields only modeled by (constant)
0-ary functions, not as a separate component in objects. Thus, names
of fields and methods in $\COOP$ objects share the same namespace.
\item Since we do not use $\COOP$ to prove type safety results (even though
it can be used), $\COOP$ does not need to have a counterpart to the
$\mathcal{W=}\{wrong\}$ domain that is used in $\SOOP$ to detect
type errors.
\end{enumerate}
When compared to $\NOOP$, as presented in~\cite{NOOP,NOOPbook,NOOPsumm,AbdelGawad14},
it is easy to see that $\COOP$, and thus also $\SOOP$, does not
accurately capture the notion of inheritance as it has evolved in
statically-typed nominal OO languages like \noun{Java}~\cite{JLS05},
\noun{C++~\cite{CPP2011}, C\#}~\cite{CSharp2007}, \noun{Scala}~\cite{Odersky09},
and \noun{X10~}\cite{X1011}.

\subsection{\label{sec:COOP-Domain-Equation}$\protect\COOP$ Domain Equation}

The domain equation that defines $\COOP$ makes use of two simple
domains $\dom B$ and $\dom L$. Domain $\dom B$ is a domain of
atomic ``base objects''. $\dom B$ could be a domain that contains
a single non-bottom value, \emph{e.g.}, $unit$ or $null$, or the
set of Boolean values \{$true$, $false$\}, the set of integers,
or some more complex set of primitive values that is the union of
Boolean values and various forms of numbers (\emph{e.g.}, whole numbers
and floats) and other primitive objects, such as characters and strings,
etc.

Domain $\dom L$ is a flat\emph{ }countable non-empty domain of labels.
Elements of $\dom L$ (or, $\left|\dom L\right|$, more accurately)
are proper labels used as names of record members (fields and methods),
or the improper ``bottom label'', $\bot_{\dom L}$, that is added
to proper labels to make $\dom L$ a flat domain. Elements of $\dom L$
other than $\bot_{\fb L}$ (proper labels) will serve as method names
in $\COOP$.

The domain equation of $\COOP$ is
\begin{equation}
\dom O\,=\,\dom B\,+\,\dom L\multimap(\dom O^{*}\strfunarr\dom O)\label{eq:core-OOP}
\end{equation}

Domain $\dom O$ is a domain of simple objects, and it is the primary
domain of $\COOP$. Equation~(\ref{eq:core-OOP}) states that a $\COOP$
object (an element of $\dom O$) is either (1) a base object (an element
of domain $\mathcal{B}$); or is (2) a record of methods (\emph{i.e.},
a finite mapping from labels, functioning as method names, to functions),
where, in turn, methods are functions from sequences of objects to
objects.

\subsection{\label{sec:COOP-Construction}$\protect\COOP$ Construction}

The construction of domain $\dom O$, as the solution of domain
equation~(\ref{eq:core-OOP}), is done using standard techniques
for solving recursive domain equations (we use the `least fixed point
(lfp) construction', which, according to Plotkin~\cite{Tomega},
is equivalent to the `inverse limit' construction).

Conceptually, the right-hand-side (RHS) of the $\COOP$ domain equation
(Equation~(\ref{eq:core-OOP})) is interpreted as a function
\begin{equation}
\lambda\dom O_{i}.\mathcal{B}\,+\,\mathcal{L}\multimap(\dom O_{i}^{*}\strfunarr\dom O_{i})\label{eq:COOP-dom-fun}
\end{equation}
over domains, from a putative interpretation $\dom O_{i}$ for $\dom O$
to a better approximation $\dom O_{i+1}$ for $\dom O$. Each element
in this sequence is a domain. The solution, $\dom O$, to the domain
equation is the least upper bound (lub) of the sequence $\dom O_{0}$,
$\dom O_{1}$, $\ldots$ .

Thus, the construction of $\dom O$ proceeds in iterations, numbered
$i+1$ for $i\geq0$. We use the empty domain as the initial value,
$\dom O_{0}$, for domain $\dom O$, and for each iteration $i+1$
we take the output domain produced by the domain constructions using
the domains $\dom O_{i}$, $\dom L$ and $\dom B$ (the values for
the function given by Formula~(\ref{eq:COOP-dom-fun})) as the domain
$\dom O_{i+1}$ introduced in iteration $i+1$.

\subsubsection{A General $\protect\COOP$ Construction Iteration}

For a general iteration $i+1$ in the construction of $\COOP$, the
construction method thus proceeds by constructing 
\[
\fb M_{i+1}=\fb O_{i}^{*}\strfunarr\fb O_{i}
\]
using the strict continuous functions domain constructor, $\strfunarr$,
and the sequences domain constructor, $^{*}$. Then, using the records
domain constructor, $\multimap$, we construct the domain of records
\[
\fb R_{i+1}=\fb L\multimap\fb M_{i+1}
\]
and, finally, using the coalesced sum domain constructor, +, we construct
\[
\fb O_{i+1}=\fb B+\fb R_{i+1}.
\]

\subsubsection{The Solution of the $\protect\COOP$ Domain Equation}

Given the continuity of all domain constructors used in the function
defined by the lambda expression~(\ref{eq:COOP-dom-fun}), and given
that composition of domain constructors preserves continuity, the
function defined by the RHS of the $\COOP$ domain equation is a continuous
function~\cite[Theorem 2.10 and Corollary 2.11]{DTMonograph}. The
least upper bound (lub) of the sequence $\dom O_{0}$, $\dom O_{1}$,
$\ldots$ of domains constructed in the construction iterations is
the least fixed point (lfp) of the function given by Formula~(\ref{eq:COOP-dom-fun}).
According to standard theorems of domain theory about the lfp of continuous
functions, the lub of the domains $\dom O_{i}$ (\emph{i.e.}, their
``limit'' domain) is simply their union, and this lub is the solution
of Equation~(\ref{eq:core-OOP}).

To complete the construction of $\COOP$, we thus construct the solution
$\dom O$ of the $\COOP$ domain equation by constructing the union
of all constructed domains $\dom O_{i}$ ,\emph{i.e.}, $\fb O$
will be given by the equation 
\[
\fb O=\bigcup_{i\geq0}\fb O_{i}.
\]

\bibliographystyle{alpha}
\bibliography{NOOPDomainTheoryCOOP}

\newcommand{\etalchar}[1]{$^{#1}$}
\begin{thebibliography}{GHK{\etalchar{+}}03}

\bibitem[Abd12]{NOOP}
Moez~A. AbdelGawad.
\newblock {\em {NOOP}: A Mathematical Model of Object-Oriented Programming}.
\newblock PhD thesis, Rice University, 2012.

\bibitem[Abd13]{NOOPbook}
Moez~A. AbdelGawad.
\newblock {\em {NOOP}: A Nominal Mathematical Model Of Object-Oriented
  Programming}.
\newblock Scholars' Press, 2013.

\bibitem[Abd14]{NOOPsumm}
Moez~A. AbdelGawad.
\newblock A domain-theoretic model of nominally-typed object-oriented
  programming.
\newblock {\em Journal of Electronic Notes in Theoretical Computer Science
  (ENTCS), DOI: 10.1016/j.entcs.2014.01.002. Also presented at The
  6\textsuperscript{th} International Symposium on Domain Theory and Its
  Applications (ISDT'13)}, 301:3--19, 2014.

\bibitem[AC14]{AbdelGawad14}
Moez~A. AbdelGawad and Robert Cartwright.
\newblock {NOOP}: A domain-theoretic model of nominally-typed object-oriented
  programming.
\newblock {\em Submitted for publication}, 2014.

\bibitem[AJ94]{Abramsky94}
Samson Abramsky and Achim Jung.
\newblock Domain theory.
\newblock In Dov M.~Gabbay S.~Abramsky and T.~S.~E. Maibaum, editors, {\em
  Handbook for Logic in Computer Science}, volume~3. Clarendon Press, 1994.

\bibitem[All86]{Allison86}
Lloyd Allison.
\newblock {\em A Practical Introduction to Denotational Semantics}.
\newblock Cambridge University Press, 1986.

\bibitem[BC04]{Bertot2004}
Yves Bertot and Pierre Casteran.
\newblock {\em Interactive Theorem Proving and Program Development Coq'Art: The
  Calculus of Inductive Constructions}.
\newblock Springer, 2004.

\bibitem[Bre58]{Breuer06}
Joseph Breuer.
\newblock {\em Introduction to the Theory of Sets}.
\newblock Dover Publications, 2006 (first published 1958).

\bibitem[CA13]{InhSubtyNWPT13}
Robert Cartwright and Moez~A. AbdelGawad.
\newblock Inheritance \emph{Is} subtyping (extended abstract).
\newblock In {\em The 25\textsuperscript{th} Nordic Workshop on Programming
  Theory (NWPT)}, Tallinn, Estonia, 2013.

\bibitem[Car84]{Cardelli84}
Luca Cardelli.
\newblock A semantics of multiple inheritance.
\newblock In {\em Proc. of the internat. symp. on semantics of data types},
  volume 173, pages 51--67. Springer-Verlag, 1984.

\bibitem[Car88]{Cardelli88}
Luca Cardelli.
\newblock A semantics of multiple inheritance.
\newblock {\em Inform. and Comput.}, 76:138--164, 1988.

\bibitem[CP88]{DTMonograph}
Robert Cartwright and Rebecca Parsons.
\newblock Domain theory: An introduction, 1988.
\newblock Monograph (based on earlier notes by Dana Scott).

\bibitem[CPP11]{CPP2011}
{\em {ISO/IEC 14882:2011: Programming Languages: C++}}.
\newblock 2011.

\bibitem[CSh07]{CSharp2007}
C\# language specification, version 3.0.
\newblock http://msdn.microsoft.com/vcsharp, 2007.

\bibitem[DP90]{DaveyPriestley90}
B.~A. Davey and H.~A. Priestley.
\newblock {\em Introduction to Lattices and Order}.
\newblock Cambridge University Press, first edition, 1990.

\bibitem[End77]{Enderton77}
Herbert~B. Enderton.
\newblock {\em Elements of Set Theory}.
\newblock Academic Press, New York, 1977.

\bibitem[GHK{\etalchar{+}}03]{Gierz2003}
G.~Gierz, K.~H. Hofmann, K.~Keimel, J.~D. Lawson, M.~W. Mislove, and D.~S.
  Scott.
\newblock {\em Continuous Lattices and Domains}, volume~93 of {\em Encyclopedia
  Of Mathematics And Its Applications}.
\newblock Cambridge University Press, 2003.

\bibitem[GJSB05]{JLS05}
James Gosling, Bill Joy, Guy Steele, and Gilad Bracha.
\newblock {\em The {J}ava Language Specification}.
\newblock Addison-Wesley, 2005.

\bibitem[GS90]{GunterHandbook90}
C.~A. Gunter and Dana~S. Scott.
\newblock {\em Handbook of Theoretical Computer Science}, chapter Semantic
  Domains.
\newblock 1990.

\bibitem[Hal60]{Halmos60}
Paul~R. Halmos.
\newblock {\em Naive Set Theory}.
\newblock D. Van Nostrand Company, Inc., 1960.

\bibitem[Har05]{OrderedSets05}
Egbert Harzheim.
\newblock {\em Ordered Sets}.
\newblock Springer, 2005.

\bibitem[KP93]{KahnConDoms93}
Gilles Kahn and Gordon~D. Plotkin.
\newblock Concrete domains, May 1993.

\bibitem[Ode09]{Odersky09}
Martin Odersky.
\newblock The scala language specification, v. 2.7.
\newblock http://www.scala-lang.org, 2009.

\bibitem[Plo78]{Tomega}
Gordon~D. Plotkin.
\newblock $\mathbb{T}^\omega$ as a universal domain.
\newblock {\em Journal of Computer and System Sciences}, 17:209--236, 1978.

\bibitem[Plo83]{PlotkinDomains83}
Gordon~D. Plotkin.
\newblock Domains.
\newblock Lecture notes in advanced domain theory, 1983.

\bibitem[SBP{\etalchar{+}}11]{X1011}
Vijay Saraswat, Bard Bloom, Igor Peshansky, Olivier Tardieu, and David Grove.
\newblock {\em X10 Language Specification: V. 2.2}, May 2011.

\bibitem[Sco76]{DTAL}
Dana~S. Scott.
\newblock Data types as lattices.
\newblock {\em SIAM Journal of Computing}, 5(3):522--587, 1976.

\bibitem[Sco81]{ScottCompMathTheoryLects}
Dana~S. Scott.
\newblock Lectures on a mathematical theory of computation.
\newblock Technical Monograph PRG-19, Oxford University Computing Laboratory,
  May 1981.

\bibitem[Sco83]{Scott82}
Dana~S. Scott.
\newblock Domains for denotational semantics.
\newblock Technical report, Computer Science Department, Carnegie Mellon
  University, 1983.

\bibitem[Sto77]{Stoy77}
Joseph~E. Stoy.
\newblock {\em Denotational Semantics: The Scott-Strachey Approach to
  Programming Language Theory}.
\newblock MIT Press, 1977.

\bibitem[TGS08]{Turbak2008}
Franklyn Turbak, David Gifford, and Mark~A. Sheldon.
\newblock {\em Design Concepts in Programming Languages}.
\newblock MIT Press, 2008.

\end{thebibliography}

\end{document}